\documentclass[11pt]{article}
\usepackage{a4,amsmath,amsfonts,amssymb,epsfig}
\usepackage{color,epsf}
\usepackage{latexsym}
\usepackage{graphicx}
\usepackage{amsthm}
\usepackage{url}
\usepackage{algorithm}
\usepackage{algpseudocode}

\usepackage{stix}

\usepackage{subfig}
\usepackage{alltt}
\usepackage{orcidlink}

\newcommand{\e}{\ensuremath{\mathrm{e}}}
\newcommand{\ad}{\ensuremath{\mathrm{ad}}}

\newtheorem{theorem}{Theorem}[section]
\newtheorem{lemma}[theorem]{Lemma}

\newtheorem{corollary}[theorem]{Corollary}

\begin{document}

\title{Error bounds for the truncated Baker--Campbell--Hausdorff and Zassenhaus formulas in unitary problems}

\author{A. Arnal$^{\orcidlink{0000-0002-3283-3379}}$, F. Casas$^{\orcidlink{0000-0002-6445-279X}}$, 
J.L. Ruiz-Benito$^{\orcidlink{0009-0002-8208-297X}}$ 
  \\[2ex]
 {\small\it Departament de Matem\`atiques and IMAC, Universitat Jaume I}, \\{\small\it 12071-Castell\'on, Spain}
}

\maketitle

\begin{abstract}
The Baker--Campbell--Hausdorff (BCH) formula plays a critical role in many branches of mathematics and physics. It expresses the logarithm of the product of exponentials of non-commuting operators as an infinite series of nested commutators of the operators involved. The Zassenhaus formula is the dual of the BCH formula: the exponential of a sum of operators is written as an infinite product of exponentials involving the operators and their commutators. 
In practical computations, however, one typically has to truncate the expansions, and so understanding the error committed by the resulting approximations
and eventually providing suitable bounds for this error is of paramount interest.  In this work we present a general strategy to derive rigorous error bounds and explicit error constants
for the BCH and Zassenhaus formulas when the operators involved are skew-adjoint, as is the case for quantum evolution problems.

\end{abstract}

\bigskip

\textbf{Keywords:} Baker--Campbell--Hausdorff theorem, Zassenhaus formula, \\
Error 
bounds, Unitary problems

\section{Introduction}
\label{sec.1}

\paragraph{Baker--Campbell--Hausdorff theorem.}
The Baker--Campbell--Hausdorff (BCH) formula is one of the most studied objects in the theory of Lie algebras and Lie groups, with profound impact 
 in other areas
of mathematics and physics, ranging from control theory to numerical analysis of differential equations and quantum mechanics  \cite{dragt16lmf,gorbatsevich97fol,bonfiglioli12tin}. 

Given two arbitrary non-commutative operators $A$ and $B$, the BCH theorem states that 
\begin{equation} \label{bch.0}
   \e^A \, \e^{B} = \e^{\Phi(A,B)},
\end{equation}   
where
\begin{equation} \label{bch.1}
  \Phi(A,B) = \log( \e^A \, \e^B) =   A + B + \sum_{n\ge 2} \Phi_n(A,B),
\end{equation}
and every
 $\Phi_n(A,B)$ is a \emph{homogeneous Lie polynomial}, that is, 
a linear combination
of commutators of the form $[V_1,[V_2, \ldots,[V_{n-1},V_n]
\ldots]]$ with $V_i \in \{A,B\}$ for $1 \le i \le n$, the
coefficients being rational constants. As
usual, the commutator $[A,B]$ is defined as $AB - BA$. 

The fact that $\Phi$ admits a representation as a series composed exclusively of commutators (or, more generally, Lie brackets) of $A$ and $B$ entails significant theoretical implications, particularly within the framework of Lie group theory. This result provides an explicit expression for the group multiplication in terms of the Lie bracket operation of the corresponding Lie algebra, and also prove the existence of a local Lie group associated with a given Lie algebra
\cite{gorbatsevich97fol}.

\paragraph{Explicit form (Dynkin).} An explicit expression for $\Phi_n$ in the BCH series \eqref{bch.1} was obtained
 by Dynkin \cite{dynkin47eot,dynkin00spo} in the form
\begin{equation}  \label{eq.1.2}
  \Phi_n(A,B) =  \sum_{p_i, q_i}
    \frac{(-1)^{n-1}}{n} \frac{[A^{p_1} B^{q_1} \ldots
      A^{p_n} B^{q_n} ]}{(\sum_{i=1}^n (p_i + q_i)) \,
      p_1! \, q_1! \ldots p_n! \, q_n!},
\end{equation} 
where the summation is taken over all non-negative
integers $p_1, q_1, \ldots$, $p_n$, $q_n$ such that
$p_i+q_i>0$, $i=1,\ldots,n$, and 
$[ A^{p_1} B^{q_1} \ldots A^{p_n} B^{q_n}]$ denotes the right-nested commutator based on the word
$ A^{p_1} B^{q_1} \ldots A^{p_n} B^{q_n}$. Thus, for instance,
\[
   [AB^2A^2 B] \equiv [A B B A A B] \equiv [A,[B,[B,[A,[A,B]]]]].
\]
Although \eqref{eq.1.2} can be used to derive the first terms, its intrinsic complexity makes the task increasingly difficult with $n$, and thus many other
procedures to construct the series have been proposed along the years (see \cite{bonfiglioli12tin,casas09aea} and references therein). 
All of them have a basic limitation, however: not all the commutators are independent, and so a rewriting process has to be carried out to express the results in terms of a basis of the free Lie algebra $\mathcal{L}(A,B)$
generated by $A$ and $B$. This process, although can be done by computer, requires nevertheless an
increasing amount of computational time and memory resources. One of the most efficient
algorithms of this type was proposed in \cite{casas09aea}, and explicit expressions of $\Phi_{n}(A,B)$ up to $n=20$ were obtained 
in terms of the classical Hall and Lyndon basis of $\mathcal{L}(A,B)$. 

For the purpose of illustration, the first terms in the classical Hall basis read explicitly
\begin{equation} \label{explic_expre}
\aligned &  \Phi_2(A,B)  =  \frac{1}{2} [A,B], \qquad\qquad  \Phi_3(A,B)  =  \frac{1}{12} \big([A,[A,B]] -   [B,[A,B]] \big), \\
& \Phi_4(A,B)  =  -\frac{1}{24}  [B, [A,[A,B]]], \\
  & \Phi_5(A,B) =   -\frac{1}{720} [A,[A,[A,[A,B]]]] - \frac{1}{180}  [B,[A,[A,[A,B]]]]  \\
  & \quad + \frac{1}{180}  [B,[B,[A,[A,B]]]]  +  \frac{1}{720}  [B,[B,[B,[A,B]]]] \\
  & \quad - \frac{1}{120}  [[A,B],[A,[A,B]]]  - \frac{1}{360}  [[A,B],[B,[A,B]]]
\endaligned
\end{equation}

\paragraph{Convergence.} While expression \eqref{eq.1.2} holds globally within a free Lie algebra, if $A$ and $B$ are elements of a normed algebra, the corresponding series of normed elements does not necessarily converge outside a neighborhood of zero. Consequently, this series cannot, in general, be employed to compute $\Phi(A,B)$ beyond the local regime. In such circumstances, a precise characterization of the series' convergence domain is of particular 
interest. Again, this subject has received considerable attention in the literature, and many different characterizations have been obtained \cite{bonfiglioli12tin}.

If $A$ and $B$ are elements of a complete normed Lie algebra $\mathfrak{g}$ with a submultiplicative norm $\| \cdot \|$ such that
\[
  \| [A,B]\| \le 2 \|A\| \, \|B\|,
\]  
then the series $\Phi = A + B + \sum_{n\ge 1} \Phi_n$ is absolutely convergent for all $A$, $B$ in the domain $D_1 \cup D_2 \in \mathfrak{g} \times \mathfrak{g}$,
with 
\begin{equation} \label{def_D1}
  D_1 = \left\{ (A,B): \, 2 \, \|A\| \le \int_{2 \|B\|}^{2 \pi} \, \frac{1}{2 + \frac{t}{2} \left( 1 - \cot \left(\frac{t}{2} \right) \right)} dt \right\},
\end{equation}
whereas $D_2$ corresponds to interchanging $A$ by $B$ in \eqref{def_D1}. Moreover, the series diverges, in general, if $\|A\| + \|B\| \ge \pi$ \cite{blanes04otc}.

\paragraph{Truncation of the BCH series.} Even when the convergence  
is guaranteed one cannot hope in general to compute all the terms of the series, 
and so for practical implementations one has to truncate the series. It makes sense then to get some estimates of the error committed by the resulting 
approximation. To be more specific, if we denote by $Z^{[m]}$ the sum of the first $m$ terms in the series of $\Phi$, 
\begin{equation} \label{bch.2}
   Z^{[m]}(A,B) =  \sum_{n=1}^m \Phi_n(A,B),  \qquad \mbox{ with } \quad \Phi_1(A,B) \equiv A+B, 
\end{equation}
then the question is how to obtain appropriate bounds on $\|\e^{A} \, \e^{B} - \e^{Z^{[m]}(A,B)} \|$, or equivalently (within the domain of convergence), on $\|\e^{\Phi(A,B)} - \e^{Z^{[m]}(A,B)} \|$.

These error bounds become 
crucial for quantifying the deviation from the true value. An example in point is the theoretical analysis and practical development of quantum simulation algorithms, especially those related with product formulas, since these errors directly impact gate complexity and depth optimization for this type of algorithms \cite{childs21tot}. 
Moreover, having bounds depending on nested commutators could be advantageous when the operators involved nearly commute.

Surprisingly, this issue has not been addressed very often in the literature, beyond the most trivial case 
$\|\e^{A} \, \e^{B} - \e^{A+B} \|$, which of course corresponds to the
well known Lie--Trotter approximation of the exponential $\e^{A+B}$ \cite{suzuki85dfo}.  Thus, the only general result we are aware of appears in 
\cite[Theorem 5.31]{bonfiglioli12tin}, and it is formulated in terms of the norms $\|A\|$, $\|B\|$ within the domain of convergence of the series.

\paragraph{Zassenhaus formula.} One can also consider the dual of the BCH formula, namely, to find operators $C_1, C_2, \ldots$ such that
\begin{equation} \label{zass_1}
  \e^{A + B} = \e^{A} \, \e^{B} \, \prod_{n \ge 2} \e^{ C_n(A,B)} = \e^{A} \, \e^{B} \, \e^{C_2(A,B)} \cdots \e^{ C_k(A,B)} \cdots
\end{equation}
and thus `disentangle' the exponential $\e^{A + B}$. The Zassenhaus formula addresses this problem and establishes that the operator
$C_k(A,B)$ in \eqref{zass_1} is a homogeneous Lie polynomial
in $A$ and $B$ of degree $k$ \cite{magnus54ote}. In particular,
\begin{equation} \label{first_terms_z}
\begin{aligned}
 & C_2(A,B) = -\frac{1}{2} [A, B] \\
 & C_3(A,B) = \frac{1}{6} [A,[A,B]]  + \frac{1}{3} [B,[A,B]]  \\
 & C_4(A,B) = -\frac{1}{24} [A,[A,[A,B]]] - \frac{1}{8} [B,[A,[A,B]]] - \frac{1}{8} [B,[B,[A,B]]].
\end{aligned}
\end{equation}
Although there are several procedures to get the terms in the expansion \eqref{zass_1} \cite{baues81cca,suzuki77otc,weyrauch09ctb}, 
the recursive algorithm proposed in \cite{casas12eco} 
is particularly simple and allows
one to determine $C_k$ up to a prescribed degree $k$ directly in terms of the minimum number of independent commutators involving $k$ operators $A$ and $B$. 

The convergence of \eqref{zass_1} has also been established when $A$ and $B$ are elements of a Banach algebra, in the sense that
\[
 \lim_{n \rightarrow \infty} = \e^A \, \e^B \, \e^{C_2} \, \cdots \, \e^{C_n} = \e^{A+B}
\]
in a certain subset of the plane $(\|A\|, \|B\|)$ \cite{suzuki77otc,bayen79otc}. 
In particular, the convergence domain contains the region $\|A\| + \|B\| < 1.054$, and extends to the points
$(\|A\|,0)$ and $(0, \|B\|)$ for arbitrarily large values of $\|A\|$ or $\|B\|$, as shown in \cite{casas12eco}. 

The Zassenhaus formula has been widely used in many different fields: periodically driven quantum systems \cite{goldman14pdq}, 
quantum nonlinear optics \cite{quesada14eot}, quantum control \cite{weidner25rqc}, neutrino oscillations \cite{chattopadhyay22npw}, 
and quantum computing \cite{zuk24rgw}, to quote just some of them. Again, for practical computations, the infinite product in \eqref{zass_1} has to be
truncated, and then the problem is to determine appropriate theoretical error bounds for the resulting approximations. To the best of our knowledge, no previous results on this issue have been published so far.

\paragraph{Main contribution.} In this work we present a general strategy that allows us to obtain rigorous bounds on the error arising when
both the BCH series and the Zassenhaus formula are truncated after an arbitrary number of terms $m \ge 1$. 
For simplicity---and due to its relevance in the simulation of quantum systems---we analyze the case where $A$ and $B$ are self-adjoint in the finite-dimensional setting. In this situation, $\mathrm{e}^{tA}$ and $\mathrm{e}^{tB}$ are unitary for $t \in \mathbb{R}$, and therefore each has norm equal to one.
Our analysis leads then to the following simple bound when the BCH series is
truncated after two terms,
\begin{equation} \label{nt3}
   \big\|\e^{A} \, \e^{B} -  \e^{A + B + \frac{1}{2}[A, B]} \big\| \le \frac{1}{4} \|[A,[A,B]]\| + \frac{1}{12} \|[B,[A,B]]\|,
\end{equation}
and the optimal estimate
\begin{equation} \label{zass_b2}
   \big\|\e^{A + B} -  \e^{A} \, \e^{B}  \,  \e^{-\frac{1}{2}[A, B]} \big\| \le \frac{1}{6} \|[A,[A,B]]\| + \frac{1}{3} \|[B,[A,B]]\|
\end{equation}
for the truncated Zassenhaus formula.
These results, as far as we know, have not been previously reported in the literature. We also generalize the treatment to the BCH series with any number of operators and also consider the symmetric BCH formula. The approach we follow leads in all cases to bounds that
are linear combinations of nested commutators of the operators involved, thus exhibiting the so-called commutator scaling property \cite{childs21tot}.

\section{Error bounds for the BCH formula with two operators}
\label{sec.2}

\subsection{Two preliminary lemmas}

In our analysis we will repeatedly use the following two results. The first one relates bounds on the difference of two time-dependent linear
unitary operators with the coefficient operators of the differential equations they satisfy.

\begin{lemma} \label{lemma_estimate}
Let $F_1(t)$, $F_2(t)$ be two unitary operators verifying the initial value problems
\begin{equation} \label{ivp_2op}
   \frac{d F_i}{dt} = M_i(t) F_i, \qquad F_i(0) = I, \qquad\qquad i=1,2,
\end{equation}
with $M_i(t)$ skew-adjoint operators. Then,
\[
\|F_1(t) - F_2(t)\| \leq 
\int_0^t \| M_1(s) - M_2(s)\| \, ds.
\]
\end{lemma}  

\begin{proof}
The difference operator $W(t)=F_1(t) - F_2(t)$ verifies
\[
  \frac{d }{dt} W(t) = M_1(t) W(t) + S(t) \qquad W(0) = 0,
\]
with $S(t) = (M_1(t) - M_2(t)) F_2(t) $. Integrating this equation leads to
\[
  W(t) = \Psi(t) \int_0^t \Psi^{-1}(s) S(s) ds,
\]
where $\Psi(t)$ is the evolution operator of the system $\frac{d }{dt} W(t) = M_1(t) W(t)$. Since both $F_2$(t) and $\Psi(t)$ are unitary, then 
$\|F_2(t)\| = 1$ and $\| \Psi(t)\| = 1$, so that
\[
  \|W(t)\| \le \int_0^t \|S(s)\| ds \le \int_0^t \| M_1(s) - M_2(s)\| \, ds.
\]  
\end{proof}  
The second result is related with a well known result in Lie theory. If we denote $\ad_X Y = [X,Y] = X Y  - Y X$, then
\begin{equation} \label{exp_ad}
  \e^{s \, \ad_{X}} Y = \e^{s \,  X} \, Y \, \e^{-s \,  X}  = \sum_{k \ge 0} \frac{s^k }{k!} \ad_X^k Y = \sum_{k \ge 0} \frac{s^k }{k!}
 \underbrace{[X, [X, \ldots, [X}_{k}, Y]]].
\end{equation}
The following 
lemma  is then a simple consequence of the Taylor expansion of the function $f(s) = \e^{s X} \, Y \, \e^{-s X}$
with remainder in integral form.
\begin{lemma} \label{l:lemma2}
    Let $X, Y$ be two linear operators such that $\e^{s X}$ is well defined for all $s \in \mathbb R$. Then, for any $n\geq 1$, 
        \[
        \e^{s \, \ad_X} Y = \sum_{k=0}^{n-1} \frac{s^k}{k!} \, \ad_X^{k} Y + s^n G_n\left(s, X, Y\right),
    \]
    where
  \begin{equation}
  \label{eq:Gn}
G_n(s, X,Y) = \int_0^1 \frac{\sigma^{n-1}}{(n-1)!} \, \e^{s \, (1-\sigma)\, \ad_X}\, \ad_X^n Y\, d\sigma.
\end{equation}
If in addition $X$ is skew-adjoint, then
\begin{equation} \label{norm_G}
  \|G_n(s,X,Y)\| \le \int_0^1 \frac{\sigma^{n-1}}{(n-1)!} \|\ad_X^n Y\| d\sigma \le \frac{1}{n!} \|\ad_X^n Y\|.
\end{equation}

\end{lemma}

\subsection{An alternative way to construct the BCH series}

It is illustrative for our ulterior error analysis to show how the successive terms of the BCH series can be obtained by considering the
differential equations satisfied by the operators involved, as in Lemma \ref{lemma_estimate}. The problem can be formulated as follows: given two
operators $A$, $B$ and a parameter $t  > 0$, we introduce 
\[
\begin{aligned}
   & F_1(t) = \e^{t A} \, \e^{t B}, \\
   &  F_2(t) = \e^{\Phi(t A, t B)}, \qquad \mbox{ with } \qquad \Phi(tA,tB) = \sum_{n \ge 1} t^n \, \Phi_n(A,B),
\end{aligned}   
\]
so that our goal is to determine the terms of the series $\Phi(tA,tB)$ such that
\begin{equation} \label{or_prob}
  \e^{t A} \, \e^{t B} = \e^{\Phi(t A, t B)}.
\end{equation}  
It is clear that $F_1(t)$ and $F_2(t)$ satisfy \eqref{ivp_2op} with
\begin{equation} \label{def_emes}
  M_1(t) = \frac{dF_1(t)}{dt} F_1^{-1}(t), \qquad \mbox{ and } \qquad   M_2(t) = \frac{dF_2(t)}{dt} F_2^{-1}(t). 
\end{equation}  
A simple calculation shows that
\[
   M_1(t) = A + \e^{t \, \ad_{A}} B, 
\]
whereas to get $M_2(t)$ we have to use the expression for the
derivative of the exponential of a time-dependent operator  \cite{blanes09tme}:
\[    
   M_2(t) = \int_0^1 \e^{x \, \ad_{\Phi}} \big(\dot{\Phi}\big) dx = \sum_{k \ge 0} \frac{1}{(k+1)!} \, \ad_{\Phi}^k \dot{\Phi}.
\]
Here $\dot{\Phi}$ denotes the derivative of $\Phi(tA, tB)$ with respect to $t$, i.e., $\dot{\Phi} = \sum_{n \ge 1} n \, t^{n-1} \Phi_n$. Now the original problem
\eqref{or_prob} can be solved by forcing that $M_1(t) = M_2(t)$. The successive terms of the series $\Phi_n$ can then be determined by expanding both
operators as power series of $t$ and equating terms at each order. Specifically, 
\[
\begin{aligned}
   M_1(t) & = A + B + \sum_{k \ge 1} \frac{t^k}{k!} \ad_A^k B = A + B + t \, \ad_A B  + \frac{t^2}{2} \ad_A^2 B + \frac{t^3}{6} \ad_A^3 B + \cdots \\
   M_2(t) & = \Phi_1 + \sum_{k\ge1} (k+1) \, t^k \, \Phi_{k+1} + \sum_{k \ge 1} \frac{1}{(k+1)!} \, \ad_{\Phi}^k \dot{\Phi} = \Phi_1 + 2 \, t \, \Phi_2  \\
   &  \quad  + t^2 \left( 3 \Phi_3 + \frac{1}{2} [\Phi_1, \Phi_2] \right) + t^3 \left( 4 \Phi_4 + [\Phi_1, \Phi_3] + 
           \frac{1}{6} [\Phi_1, [\Phi_1, \Phi_2]] \right) + \cdots
\end{aligned}
\]  
In general, at order $t^{n-1}$,  we have
\begin{equation} \label{cons_terms}
  \frac{1}{(n-1)} \ad_A^{n-1} B = n \, \Phi_n + \mathcal{F}_{n-1}(\Phi_1, \ldots, \Phi_{n-1}),
\end{equation}
with $\mathcal{F}_{n-1}$ a function involving only commutators, 
whence it is possible to get $\Phi_n$ for $n \ge 2$  recursively in terms of commutators, starting with $\Phi_1 = A + B$. 

Suppose now we stop this procedure at a given index $m$, so that we end up with 
\begin{equation} \label{def_Z}   
   Z^{[m]}(tA, tB) \equiv \sum_{j=1}^m t^j \, \Phi_j(A,B),
\end{equation}
where the expression of $\Phi_j$, $j=1, \ldots, m$ has been obtained from \eqref{cons_terms}. Then, the previous analysis clearly shows that
$M_1(t) - M_2(t) = \mathcal{O}(t^m)$, and therefore  Lemma \ref{lemma_estimate} allows us to conclude that
 \begin{equation} \label{bigOtm}
  \| \e^{t A} \, \e^{t B} - \e^{Z^{[m]}(tA, tB) } \| = \mathcal{O}(t^{m+1}),
\end{equation}
where the constant implicit in \eqref{bigOtm} depends on the remaining terms of the series.

\subsection{Error bounds}

We are now interested in providing explicit error bounds when the BCH series is truncated, so as to render the constant involved in the estimate
\eqref{bigOtm} fully explicit. Thus, the problem we are addressing can be formulated as follows.

\

\noindent
\textbf{Problem 1:} \emph{Given the finite sum $Z^{[m]}(tA, tB)$ defined in \eqref{def_Z}, obtain rigorous bounds for
 \[
     \left\| \e^{t A} \, \e^{t B} - \e^{Z^{[m]}(tA, tB) } \right\| 
 \]
 in terms of commutators of $A$ and $B$, and in particular determine the constant implicit in  \eqref{bigOtm}}. 
  
 \
 
 The strategy we follow is similar to the one applied in the previous subsection: we expand the operators $M_1(t)$ and $M_2(t)$
 associated with $ \e^{t A} \, \e^{t B}$ and $\e^{Z^{[m]}(tA, tB) }$, 
 \begin{equation} \label{M1M2_trunc}
   M_1(t) = A + \e^{t \, \ad_{A}} B, \qquad\qquad  M_2(t) = \int_0^1 \e^{x \, \ad_{Z^{[m]}}} \big(\dot{Z}^{[m]}\big) dx,
 \end{equation}
 respectively, in powers of $t$ up to degree $m$ with remainder in integral form,
then get appropriate bounds for the difference $M_1(t) - M_2(t)$ and finally use Lemma \ref{lemma_estimate}. The simplest case $m=1$ illustrates the
 general procedure.
 
 We have
 \[
 \begin{aligned}
  & M_1(s) = A + B + s \int_0^1 \e^{ s(1-\sigma)\, \ad_{A}} \ad_A B \, d\sigma, \\
  & M_2(s) = \int_0^1 \e^{x \, \ad_{Z^{[1]}}} \big(\dot{Z}^{[1]}\big) dx = \int_0^1 \dot{Z}^{[1]} dx = A + B,
 \end{aligned}
 \]
 so that
 \[ 
  \|M_1(s) - M_2(s)\| =  \left\|  s \int_0^1 \e^{ s(1-\sigma)\, \ad_{A}} \ad_A B \, d\sigma \right\| \le s \int_0^1 \|  \ad_A B \| \, d \sigma = s \|  \ad_A B \|,
\]
and Lemma \ref{lemma_estimate} leads finally to the well known result for the Lie--Trotter approximation \cite{suzuki85dfo,iserles24aea}:
\begin{equation} \label{lie_trotter}
  \| \e^{t A} \, \e^{t B} - \e^{t(A+B)}\| \le \frac{t^2}{2} \, \|[A,B]\|.
\end{equation}  

\

In the general case we will apply  Lemma \ref{l:lemma2} by truncating the expansions of $M_1$ and $M_2$ in such a way that
$M_1(s) - M_2(s) = \mathcal{O}(s^m)$. The outcome is the following theorem.

\begin{theorem}\label{th:general_bound} 
Let $A$ and $B$ be skew-adjoint operators, and $Z^{[m]}(tA, tB)$ denote the sum of the first $m \ge 2$ 
terms of the Baker--Campbell--Hausdorff series and $t > 0$. Then
\[
   \left\| \e^{t A} \, \e^{t B} - \e^{Z^{[m]}(tA, tB) } \right\| \le \int_0^t \|  \mathcal{R}^{[m]}(s) \| ds + \frac{t^{m+1}}{(m+1)!} \|\ad_A^m B\| + 
   \frac{1}{m!} \int_0^t \|\ad_{Z^{[m]}}^{m-1} \dot{Z}^{[m]} \| ds
\]
where 
\begin{equation} \label{def_R}
    \mathcal{R}^{[m]}(s) = A + \sum_{k=0}^{m-1} \frac{s^k}{k!} \ad_A^k B - \sum_{k=0}^{m-2} \frac{1}{(k+1)!} \, \ad_{Z^{[m]}}^k \dot{Z}^{[m]} ,
\end{equation}
and $ \dot{Z}^{[m]}$ denotes the derivative of $Z^{[m]}(sA,sB)$ with respect to $s$.
\end{theorem}   
   
\begin{proof}
Application of Lemma \ref{l:lemma2} to \eqref{M1M2_trunc} gives
\[
 \begin{aligned}
  & M_1(s) = A + \sum_{k=0}^{m-1} \frac{s^k}{k!} \ad_A^k B + s^m G_m(s,A,B) \\
  & M_2(s) =  \sum_{k=0}^{m-2} \frac{1}{(k+1)!} \ad_{Z^{[m]}}^k \dot{Z}^{[m]} + \int_0^1 x^{m-1} G_{m-1}(x,Z^{[m]},  \dot{Z}^{[m]}) dx,
 \end{aligned}
\]
so that
\[
 \|M_1(s) - M_2(s)\| \le \|  \mathcal{R}^{[m]}(s) \| + s^m \|G_m(s,A,B)\| +  \int_0^1 x^{m-1} \| G_{m-1}(x,Z^{[m]},  \dot{Z}^{[m]})\| dx.
\]
Notice that $\mathcal{R}^{[m]}(s)$ defined by \eqref{def_R} verifies 
\[
   \mathcal{R}^{[m]}(s) = \mathcal{O}(s^m),
\]
since all terms up to $s^{m-1}$ vanish due to the way the successive terms $\Phi_1, \ldots, \Phi_m$ in the BCH series are obtained.  In fact, 
the term in $s^m$ comes from
$\sum_{k=1}^{m-2} \frac{1}{(k+1)!} \, \ad_{Z^{[m]}}^k \dot{Z}^{[m]}$.

Now, from \eqref{norm_G}, we have
  $\|G_m(s,A,B)\| \le  \frac{1}{m!} \|\ad_A^m B\|$
and
\begin{equation} \label{norm_G2}
\begin{aligned}
   & \int_0^1 x^{m-1} \| G_{m-1}(x,Z^{[m]},  \dot{Z}^{[m]})\| dx   \le  \int_0^1 \frac{x^{m-1}}{(m-1)!} 
    \| \ad_{Z^{[m]}}^{m-1} \dot{Z}^{[m]} \| \, dx   \\
  & \qquad\quad = \frac{1}{m!} \| \ad_{Z^{[m]}}^{m-1} \dot{Z}^{[m]} \|.
\end{aligned}    
\end{equation}  
Lemma \ref{lemma_estimate} leads finally to the desired result. 
\end{proof}   

If we specify Theorem \ref{th:general_bound} to the case $m=2$, we have
\[
 Z^{[2]}(sA,sB)) = s(A+B) + \frac{s^2}{2} [A,B], \qquad \dot{Z}^{[2]}(sA,sB)) = A+ B + s \, [A, B],
\]
so that
\[
  \ad_{Z^{[2]}} \big( \dot{Z}^{[2]} \big) = \frac{s^2}{2} \big(  [A,[A,B]] + [B,[A,B]] \big).
\]
In consequence,  $\mathcal{R}^{[2]}(s) = 0$ and
\begin{equation} \label{new_b}
   \left\| \e^{t A} \, \e^{t B} - \e^{Z^{[2]}(tA, tB) } \right\|  \le \frac{t^3}{6} \|\ad_A^2 B\| +  \frac{t^3}{12} 
     \big\|[A,[A,B]] +  [B,[A,B]] \big\|. 
\end{equation}
If we apply the triangular inequality to the last term in \eqref{new_b}, we finally get 
\begin{equation} \label{new_b_2}
   \left\| \e^{t A} \, \e^{t B} - \e^{Z^{[2]}(tA, tB) } \right\|  \le \frac{t^3}{4} \|[A,[A,B]]\| + \frac{t^3}{12} \|[B,[A,B]]\|,
\end{equation}     
thus recovering \eqref{nt3} with $t=1$. 

Theorem \ref{th:general_bound} gives the following structure for the error bound.

\begin{corollary}
With the same hypothesis as in Theorem \ref{th:general_bound}, it follows that
\begin{equation} \label{bound_pol}
   \left\| \e^{t A} \, \e^{t B} - \e^{Z^{[m]}(tA, tB) } \right\| \le \sum_{j=m+1}^{m^2 -1} t^j \, C_j^{[m]}(A,B),
\end{equation}
where $C_j^{[m]}(A,B)$ are linear combinations of norms of nested commutators involving $j$ operators $A$ and $B$.
\end{corollary}

\begin{proof}
Starting from  $Z^{[m]}(sA,sB) = s \Phi_1 + \cdots + s^m \Phi_m$, a straightforward calculation shows that
\[
  \ad_{Z^{[m]}}^{k} \dot{Z}^{[m]} = \sum_{j = k+1}^{(k+1)m-2} f_j^{(k)} \, s^j, \qquad\quad k \ge 1.
\]
Here $f_j^{(k)}$ denote nested commutators involving $k+1$ operators from $\Phi_1, \ldots, \Phi_m$ (but $j+1$ operators
$A$, $B$). In consequence,
\begin{equation} \label{exp_Zm}
  \ad_{Z^{[m]}}^{m-1} \dot{Z}^{[m]} = \sum_{j = m}^{m^2-2} f_j^{(m-1)} \, s^j.
\end{equation}
On the other hand, from \eqref{def_R}, and recalling that $ \mathcal{R}^{[m]}(s) = \mathcal{O}(s^m)$, application of the previous
result leads to
\begin{equation} \label{exp_R}
   \mathcal{R}^{[m]}(s) = \sum_{j=m}^{(m-1)m -2} r_j \, s^j 
\end{equation}
for some coefficients $r_j$ which are again nested commutators involving $\Phi_1, \Phi_2, \Phi_3, \ldots$. 
Application of  Theorem \ref{th:general_bound} leads then to the
desired result.
\end{proof} 

In fact, we can get more insight into the bound \eqref{bound_pol} by computing explicitly the coefficients involved. Thus, in \eqref{exp_Zm}
it holds that 
\begin{equation} \label{expre_f}
  f_j^{(m-1)} = \ \sum_{(i_1,i_2,\ldots,i_m)} \  i_m \, \ad_{\Phi_{i_1}}  \ad_{\Phi_{i_2}}  \cdots  \, \ad_{\Phi_{i_{m-1}}}  \Phi_{i_m},
\end{equation}
where the sum is extended over all partitions $(i_1,i_2,\ldots,i_m)$ of the integer $j+1$ of length $m$, 
so that $i_1 + i_2 + \cdots +i_m = j+1$, and $i_{m-1} \ne i_m$ (otherwise the commutator vanishes). 
An analogous expression follows
for the $r_j$ coefficients in \eqref{exp_R}, namely
\begin{equation} \label{expre_r}
  r_j = \sum_{\ell = 2}^{m-1} \frac{1}{\ell!} \ \sum_{(i_1,i_2,\ldots,i_{\ell})} \  i_{\ell} \, 
  \ad_{\Phi_{i_1}}  \ad_{\Phi_{i_2}}  \cdots  \, \ad_{\Phi_{i_{\ell-1}}}  \Phi_{i_{\ell}}.
\end{equation}
Adding up all the contributions according with Theorem \ref{th:general_bound} and carrying out the integration in $s$, we arrive at
\begin{equation} \label{ref_b}
\begin{aligned}
    \left\| \e^{t A} \, \e^{t B} - \e^{Z^{[m]}(tA, tB) } \right\| & \le \left\| \sum_{j = m+1}^{(m-1)m-1} \frac{1}{j}  r_{j-1} \, t^j \right\| + 
    \frac{1}{(m+1)!} \| \ad_A^m B\| \, t^{m+1} \\
    & +  \frac{1}{m!} \left\| \sum_{j = m+1}^{m^2-1} \frac{1}{j}  \, f_{j-1}^{(m-1)} \, t^j \right\|.
\end{aligned}    
\end{equation}    
The explicit expression of $C_j^{[m]}$ are obtained by applying the triangle inequality to \eqref{ref_b}, then by inserting
 the expressions of $\Phi_1, \ldots, \Phi_m$ in terms of $A$, $B$ 
into \eqref{expre_f} and \eqref{expre_r}, carrying out the computations involved and then writing the result
 in a basis of the free Lie algebra $\mathcal{L}(A,B)$. The final bound is then
a polynomial in the parameter \(t\) of degree \(m^2 - 1\), whose lowest-order term is \(t^{m+1}\). Only when $m=2$ does the expansion reduce 
 to a single term proportional to \(t^3\). The whole procedure can be  implemented with a
computer algebra system, although the computational complexity grows exponentially with $m$. With a \emph{Mathematica} code, we have obtained
the explicit expression of \eqref{ref_b} when $m=3$ (collected in the Appendix), whereas the bound \eqref{bound_pol} for $m=4$ is available at the website 
\begin{center}
  \url{http://www.gicas.uji.es/Research/bch.html}
\end{center}  
in the classical Hall
basis of the free Lie algebra $\mathcal{L}(A,B)$. For completeness, we have also included similar results in terms of independent right-nested commutators.

Notice that the first term $C_{m+1}^{[m]}$ in the error bound \eqref{bound_pol}, i.e., the main term in the error expansion, is particularly easy to compute, since 
\[
  f_{m}^{(m-1)} = \ad_{\Phi_1}^{m-1} \Phi_2 = \frac{1}{2} \ad_{A+B}^{m-1} \, \ad_A B,
\]  
so that
\[
  C_{m+1}^{[m]} = \frac{1}{m+1} \|r_m\| + \frac{1}{(m+1)!} \left( \| \ad_A^m B\| + \frac{1}{2} \|\ad_{A+B}^{m-1} \, \ad_A B\| \right).
\]
For completeness, we also collect at the same website the expression of $C_{m+1}^{[m]}$
for $m=3,\ldots,7$.

\section{An illustrative example}

To illustrate the sharpness of the previous bounds, we consider the following simple example. 
We first generate two $20 \times 20$ complex matrices whose entries (both real and imaginary parts) 
are drawn from the standard normal distribution. From these, we construct skew-adjoint matrices 
in the usual way, compute their spectral norms, and then rescale them to obtain matrices $A$ and $B$ 
satisfying $\|A\| = \|B\| = 1$.

Next, we select 23 values of $t$ in the interval $[0.001, 4]$ and, for each value, evaluate the matrices 
$\e^{tA}\e^{tB}$ and $\e^{Z^{[2]}}$. We then compute
\begin{equation} \label{act_err_2}
E_r^{[2]} \equiv \|\e^{tA}\e^{tB} - \e^{Z^{[2]}}\|
\end{equation}
and compare it with the bounds
\begin{equation} \label{bounds_2_3}
\begin{aligned}
  B_1^{[2]} &\equiv \frac{t^3}{6} \|[A,[A,B]] \| 
  + \frac{t^3}{12} \big\|[A,[A,B]] + [B,[A,B]]\big\|, \\
  B_2^{[2]} &\equiv \frac{t^3}{4} \|[A,[A,B]]\| 
  + \frac{t^3}{12} \|[B,[A,B]]\|.
\end{aligned}
\end{equation}
Notice that $B_2^{[2]}$ is obtained by applying the triangular inequality to $B_1^{[2]}$.
This procedure yields Figure~\ref{figure.1}. Observe that the $t^2$ scaling of $E_r^{[2]}$ 
persists even for relatively large errors. Moreover, $B_1^{[2]}$ and $B_2^{[2]}$ differ only slightly, 
and both remain very close to the exact result until the actual error is of the order of 1. In particular, 
the ratio $E_r^{[2]} / B_1^{[2]}$ ranges from $0.369$ for $t = 0.001$ to $0.388$ for $t = 1$. 
Repeating the experiment over 300 additional random trials yields an average value of approximately 
$0.427$ for this ratio at $t = 1$.

\begin{figure}[!ht] 
\centering
  \includegraphics[width=.8\textwidth]{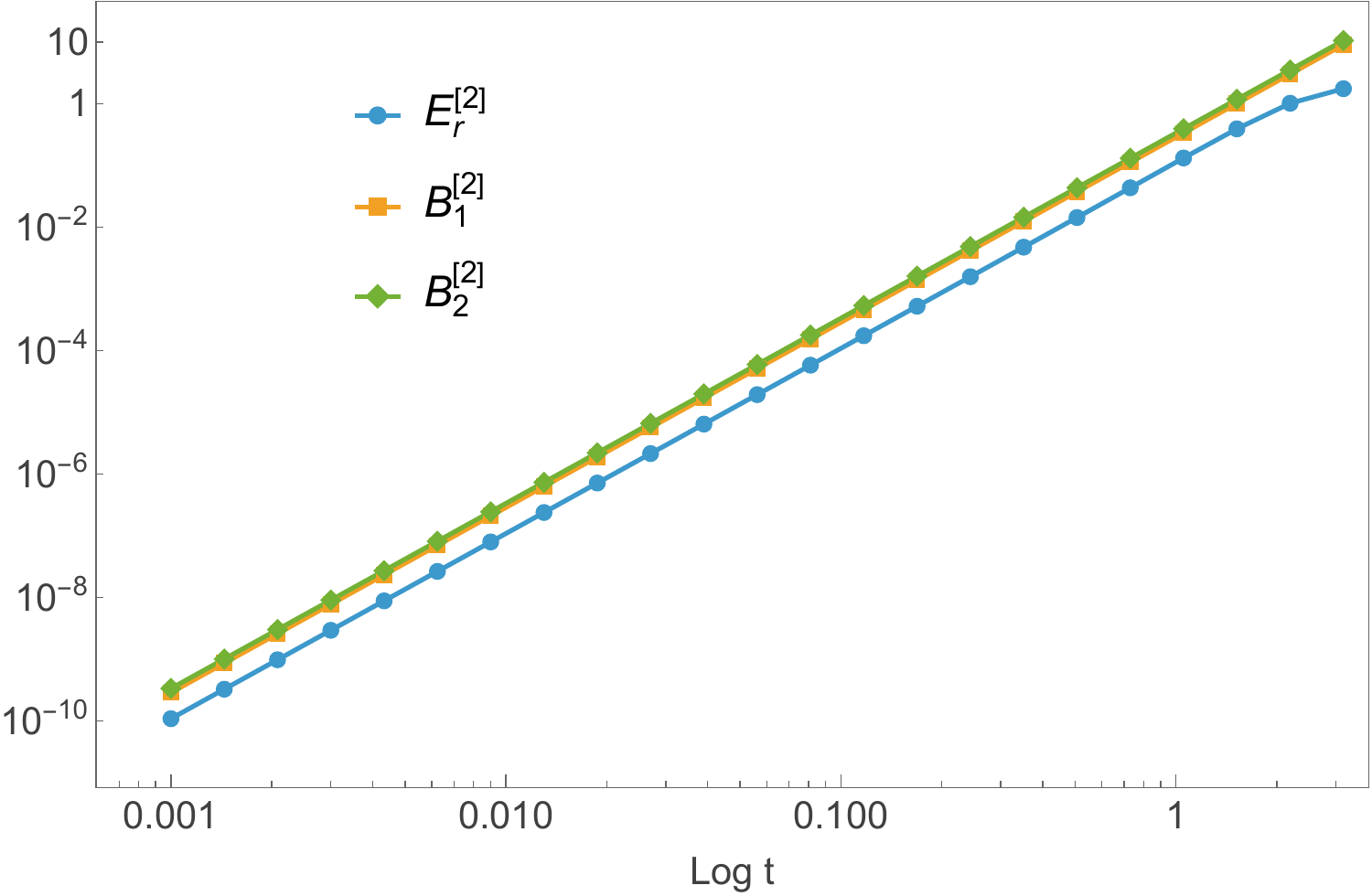}
\caption{\label{figure.1} \small Actual error \eqref{act_err_2} committed by the truncated second-order BCH formula, together with the bounds 
\eqref{bounds_2_3} for two skew-adjoint $20 \times 20$ matrices $A$ and $B$. The difference between $B_1^{[2]}$ and $B_2^{[2]}$ is almost negligible in
this case, and they overestimate the actual error by less than a factor 3.
}
\end{figure}

Figure \ref{figure.2} corresponds to the truncated BCH formula after $m=3$ terms. Specifically, we represent the actual error 
$E_r^{[3]} \equiv  \|\e^{t A} \e^{t B} - \e^{Z^{[3]}}\|$ as a function of $t$, together with the bound $B_1^{[3]}$ collected in the Appendix, eq.
\eqref{B1_3}, and the corresponding expression $B_2^{[3]}$ obtained from $B_1^{[3]}$ by applying the triangle inequality to each term $C_j^{[3]}$. Again, the
difference between $B_1^{[3]}$ and $B_2^{[3]}$ is almost negligible, and both of them differ from the exact result by an order of magnitude.

\begin{figure}[!ht] 
\centering
  \includegraphics[width=.8\textwidth]{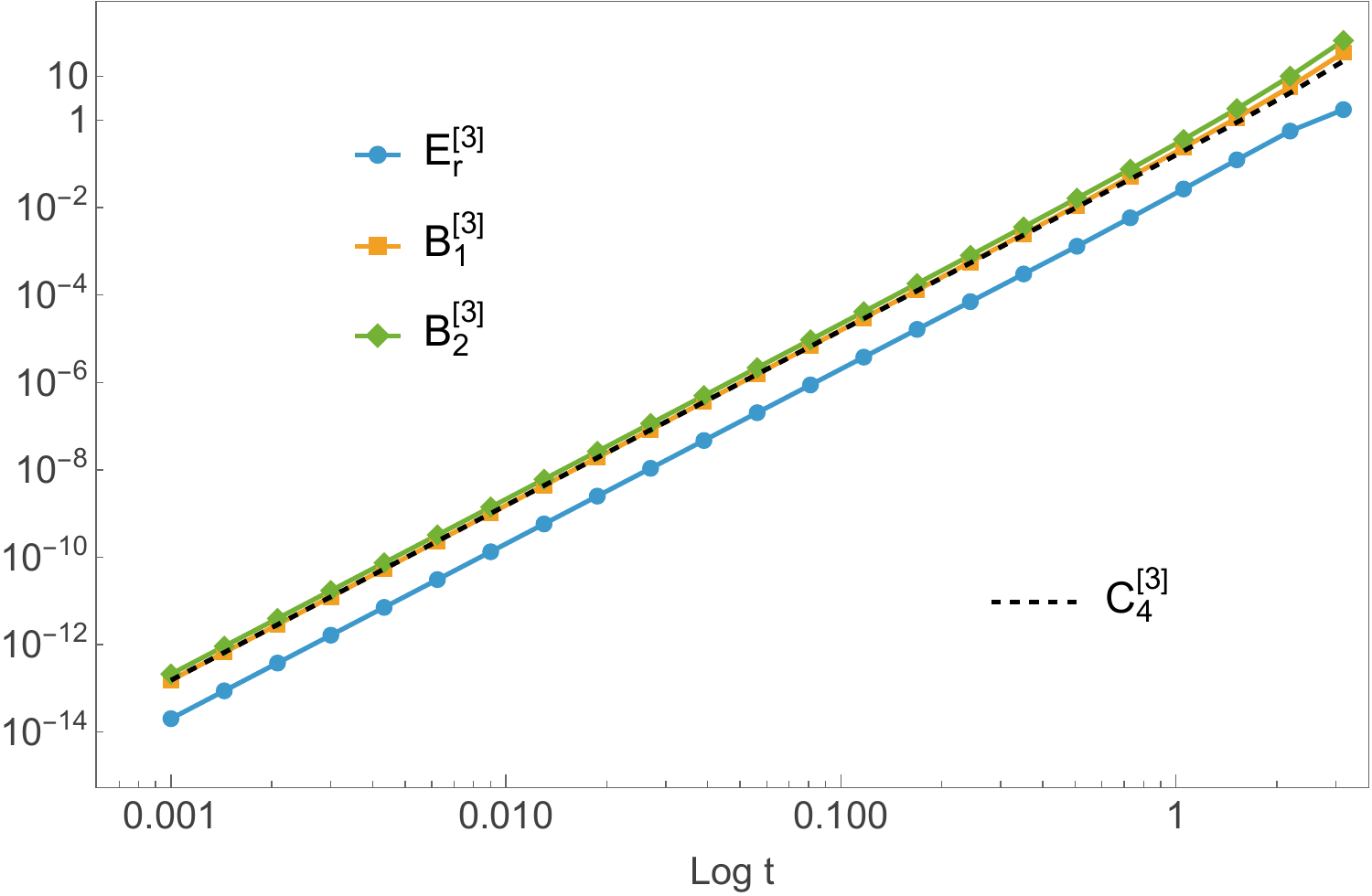}
\caption{\label{figure.2} \small Actual error $E_r^{[3]}$ committed by the truncated second-order BCH formula, together with the theoretical bounds 
obtained in this work for two skew-adjoint $20 \times 20$ matrices $A$ and $B$.  The dashed line corresponds to the leading term in the theoretical bound
\eqref{bound_pol}.
}
\end{figure}

We have also included the result obtained by retaining only the leading term, $t^4 C_4^{[3]}$, in the polynomial bound \eqref{bound_pol}. As can be seen from the figure, this dominant contribution already provides a remarkably accurate approximation to the theoretical bound. Including additional terms does not significantly deteriorate the estimate unless $t$ takes large values, where the bound itself becomes irrelevant. 
A similar behavior is observed for larger values of $m$.

\section{The truncated BCH formula with $N > 2$ operators}
\label{general_case}

Consider now the BCH formula involving any number of skew-adjoint matrices, $A_1, \ldots$, $A_N$,
\[
  \exp(t A_N) \, \cdots \, \exp(t A_2) \, \exp(t A_1) = \exp\big(\Phi(A_N, \ldots,  A_2,  A_1)\big),
\]  
where
\begin{equation} \label{bch_series_g}
  \Phi(t A_N, \ldots, t A_2, t A_1) = \sum_{j \ge 1} t^j \, \Phi_j(A_N, \ldots, A_2, A_1)
\end{equation}
and $\Phi_n$ can be written as a linear combination of commutators involving $n$ operators from the set $\{A_1, \ldots, A_N\}$.  Explicit expressions
of $\Phi_n$ can be obtained in terms of products of the form $A_{\sigma(1)} A_{\sigma(2)} \cdots A_{\sigma(n)}$ for all permutations $\sigma$ of
$\{1, 2, \ldots, n\}$ (see \cite{burgunder08eia,dynkin49otr,loday94sdh} for a detailed treatment) and then applying either the Dynkin--Specht--Weber theorem 
\cite{jacobson79la} or the algorithm presented in \cite{arnal21ano}.

The previous procedure can also be applied here when the series \eqref{bch_series_g} is truncated after, say, $m$ terms to get 
 a bound for $\| \e^{t A_N} \, \cdots \, \e^{t A_2} \, \e^{t A_1} - \e^{Z^{[m]}(t)} \|$, with
\begin{equation} \label{new_Z}
  Z^{[m]}(t) = \sum_{j=1}^m t^j \, \Phi_j(A_N, \ldots, A_1).
\end{equation}  
According with the proof of Theorem \ref{th:general_bound}, we have to determine the operators $M_1(t)$ and $M_2(t)$ related with
\[
  F_1(t) =  \e^{t A_N} \, \cdots \, \e^{t A_2} \, \e^{t A_1}  \qquad \mbox{ and } \qquad F_2(t) = \exp \big(Z^{[m]}(t) \big),
\]
respectively, via \eqref{def_emes}. In fact, a recursive procedure to get  $M_1(t)$ has been recently obtained in \cite{casas26ebf} in the context
of the error analysis of splitting methods. Specifically,
\[
 M_1(t) =  \sum_{n=0}^{m-1} t^n \, C_{N,n} + t^m \, V_{N,m}(t),
\]
where the terms $C_{N,n}$ are determined through  
\begin{equation} \label{recur_C}
\begin{aligned}
 & C_{k,0} = A_1 + \cdots + A_k, \qquad 1 \le k \le N \\
 & C_{1,n} = 0, \qquad n \ge 1 \\
 & C_{k,n} = \sum_{j=0}^n \frac{1}{j!} \,\ad_{A_k}^j C_{k-1, n-j},\qquad 2 \le k \le N, \;\; n\geq 1,
\end{aligned}
\end{equation}
whereas the remainder is computed from
\[
\begin{aligned}
  & V_{1,m}(t) = 0 \\
  & V_{k,m}(t) = \sum_{n=1}^m G_n(t, A_k, C_{k-1,m-n}) + \e^{t \, \ad_{A_k}} V_{k-1,m}(t), \qquad k \ge 2.
\end{aligned}
\]  
Taking into account the expression of $G_n(t, A_k, C_{k-1,m-n})$ and the bound \eqref{norm_G}, we get
\[
  \| V_{k,m}(t) \| \le  \int_0^1 \left\| \sum_{n=1}^m \frac{\sigma^{n-1}}{(n-1)!} \, \ad_{A_k}^n C_{k-1,m-n}  \right\| d \sigma+ \|V_{k-1,m}(t)\|,
\]  
so that
\begin{equation} \label{norm_V_N}
   \| V_{N,m}(t) \| \le \sum_{k=2}^N   \int_0^1 \left\| \sum_{n=1}^m \frac{\sigma^{n-1}}{(n-1)!} \, \ad_{A_k}^n C_{k-1,m-n}  \right\| d \sigma.
\end{equation}   
With respect to $M_2(t)$, we have as before
\[
  M_2(t) =  \sum_{k=0}^{m-2} \frac{1}{(k+1)!} \ad_{Z^{[m]}}^k \dot{Z}^{[m]} + \int_0^1 x^{m-1} G_{m-1}(x,Z^{[m]},  \dot{Z}^{[m]}) dx,
\]
with $Z^{[m]}$  given by \eqref{new_Z}. In consequence,
\[
  \| M_1(s) - M_2(s)\| \le \|\mathcal{R}^{[m]}(s)\| + s^m \|V_{N,m}(s)\| + \int_0^1 x^{m-1} \| G_{m-1}(x,Z^{[m]},  \dot{Z}^{[m]})\| dx,
\]
with
\begin{equation} \label{def_r_2}
      \mathcal{R}^{[m]}(s) =  \sum_{k=0}^{m-1} s^k \, C_{N,k} - \sum_{k=0}^{m-2} \frac{1}{(k+1)!} \, \ad_{Z^{[m]}}^k \dot{Z}^{[m]}.
\end{equation}      
Finally, Lemma \ref{lemma_estimate} together with eqs. \eqref{norm_G} and  \eqref{norm_V_N} leads to the following result.

\begin{theorem}\label{th:Ng2} 
Let $A_1, A_2, \ldots, A_N$ be $N$ skew-adjoint operators and $Z^{[m]}(t)$ denote the sum of the first $m \ge 2$ 
terms of the Baker--Campbell--Hausdorff series \eqref{bch_series_g} with $t > 0$. Then
\[
\begin{aligned}
    \left\| \e^{t A_N} \, \cdots \, \e^{t A_1} - \e^{Z^{[m]}(t)} \right\| & \le \int_0^t \|  \mathcal{R}^{[m]}(s) \| ds + 
    \frac{1}{m!} \int_0^t \|\ad_{Z^{[m]}}^{m-1} \dot{Z}^{[m]} \| ds  \\
    & \;\; + \frac{t^{m+1}}{m+1} \, \sum_{k=2}^N   \int_0^1 \left\| \sum_{n=1}^m \frac{\sigma^{n-1}}{(n-1)!} \, \ad_{A_k}^n C_{k-1,m-n}  \right\| d \sigma,
\end{aligned}   
\]
where $\mathcal{R}^{[m]}(s)$ is given by \eqref{def_r_2} and $C_{k,n}$ are obtained with the recurrence \eqref{recur_C}.
\end{theorem}

Let us particularize this result to the case of $N=2$ operators. Then, a straightforward calculation shows that
\[
\begin{aligned}
 &   \int_0^1 \left\| \sum_{n=1}^m \frac{\sigma^{n-1}}{(n-1)!} \, \ad_{A_2}^n C_{1,m-n}  \right\| d \sigma = 
   \int_0^1  \left\| \frac{\sigma^{m-1}}{(m-1)!} \ad_{A_2}^m C_{1,0} \right\| d\sigma \\
 & \qquad   =   \int_0^1  \frac{\sigma^{m-1}}{(m-1)!} \ \| \ad_{A_2}^m A_1\| d\sigma =
  \frac{1}{m!}  \|\ad_{A_2}^m A_1  \|, 
\end{aligned}  
\]
and Theorem \ref{th:general_bound} is recovered with the replacements $A_1 \longmapsto B$ and $A_2 \longmapsto A$.

\section{Symmetric BCH formula}

Very often in practical problems one has to deal with the so-called symmetric BCH formula
\begin{equation}  \label{sBCH1}
  \exp \left(\frac{t}{2} A\right) \, \exp(t B) \, \exp\left(\frac{t}{2} A\right) = \exp \big(W(tA, tB)\big).
\end{equation}
This is the case, in particular, for palindromic splitting methods \cite{blanes24smf}.
Two applications of the usual BCH formula give  the expression of $W(t A, t B)$. More efficient procedures exist, however, 
 to construct explicitly the series $\sum_{n \ge 1} t^n \, W_n(A,B)$ defining $W$ in terms of independent commutators 
involving $A$ and $B$ up to an arbitrarily high degree. In general, $W_{2j} = 0$ for $j \ge 1$, whereas terms up to $W_{19}$
in both Hall and Lyndon bases have been obtained explicitly with the algorithm described in  \cite{casas09aea}. Thus, in particular
\[
\aligned
 &  W_1(A,B) = A + B,  \\
 &  W_3(A,B) =  -\frac{1}{24} [A,[A,B]] - \frac{1}{12} [B,[A,B]],  \\
 & W_5(A,B) = \frac{7}{5760} [A,[A,[A,[A,B]]]] + \frac{7}{1440} [B,[A,[A,[A,B]]]]  \\
 & \quad + \frac{1}{180} [B,[B,[A,[A,B]]]] + \frac{1}{720} [B,[B,[B,[A,B]]]]  \\
 &  \quad  + \frac{1}{480} [[A,B],[A,[A,B]]]  - \frac{1}{360} [[A,B],[B,[A,B]]].  \\
\endaligned   
\]

Analogously to the previous case, we are interested in obtaining bounds for the approximation that results when the series $W(tA,tB)$ in \eqref{sBCH1}
is truncated after the first $p=1,3,\ldots$ terms:
\[
   \left\| \e^{\frac{t}{2} A} \, \e^{t B} \,  \e^{\frac{t}{2} A} - \e^{W^{[p]}(tA,tB)} \right\|, 
\]   
when $A$ and $B$ are skew-adjoint operators.
It turns out that the general treatment carried out in section \ref{general_case} can be applied here with $N=3$ and operators
\[
  A_1 = A_3 = \frac{1}{2} A, \qquad A_2 = B.
\]  
However, to get better estimates,  it is advantageous to define $W_{2j} = 0$ for $j \ge 1$ and consider the truncation
\begin{equation} \label{def_W}
  W^{[p+1]}(tA,tB)= \sum_{j=1}^{p+1} t^j \, W_p(A,B) = t W_1(A,B) + t^3 W_3(A,B) + \cdots + t^p W_p(A,B).
\end{equation}  
The process can be illustrated in the simplest case $p=1$. Then
$W^{[2]} = t(A+B)$, the recursion \eqref{recur_C} gives
\[
  C_{1,0} = \frac{1}{2} A, \quad C_{1,1} = 0, \quad C_{2,0} = \frac{1}{2} A + B, \quad C_{2,1} = \frac{1}{2}\ad_B A, \quad C_{3,0} = A + B, \quad C_{3,1} = 0,
\]
so that $\mathcal{R}^{[2]}(s) = 0$, $\ad_{W^{[2]}} \dot{W}^{[2]} = 0$, and
\[
  \sum_{k=2}^3 \int_0^1 \left\| \sum_{n=1}^2 \frac{\sigma^{n-1}}{(n-1)!} \ad_{A_k}^n C_{k-1,2-n} \right\| d\sigma = 
  \frac{1}{8} \| \ad_A^2 B \| + \frac{1}{4} \|\ad_B^2 A\|. 
\]
Theorem \ref{th:Ng2}, with the obvious identification $Z^{[m]} \equiv W^{[m]}$, then leads to   
\[
  \left\| \e^{\frac{t}{2} A} \, \e^{t B} \,  \e^{\frac{t}{2} A} - \e^{t(A+B)} \right\| \le  \frac{t^3}{24} \|[A,[A,B]]\| + \frac{t^3}{12} \|[B,[A,B]]\|,
\]
thus reproducing the optimal result for the Strang splitting \cite{suzuki85dfo,childs21tot,iserles24aea}. More generally,
application of Theorem \ref{th:Ng2} with $N=3$, $Z^{[m]} = W^{[p+1]}$ and 
$A_1 = A_3 = A/2$, $A_2 = B$ allows one to get a  bound for the truncated symmetric BCH formula with $p$ terms:

\begin{theorem}\label{th:sym} 
Let $A$ and $B$ be two skew-adjoint operators and $W^{[p+1]}(t)$ denote the sum \eqref{def_W} comprising the first $p=1,3,\ldots$ 
terms of the symmetric Baker--Campbell--Hausdorff series, with $t > 0$. Then
\[
\begin{aligned}
  &   \left\| \e^{\frac{t}{2} A} \, \e^{t B} \,  \e^{\frac{t}{2} A} - \e^{W^{[p]}(tA,tB)} \right\|  \le \int_0^t \|  \mathcal{R}^{[p+1]}(s) \| ds + 
    \frac{1}{(p+1)!} \int_0^t \|\ad_{W^{[p+1]}}^{p} \dot{W}^{[p+1]} \| ds  \\
    & \qquad + \frac{t^{p+2}}{p+2} \, \left( \frac{1}{2 (p+1)!} \| \ad_B^{p+1} A\| + \int_0^1 \left\| \sum_{n=1}^{p+1} \frac{\sigma^{n-1}}{(n-1)!} 
    \ad_{A/2}^n  \, C_{2,p+1-n} \right\| d\sigma \right),
\end{aligned}   
\]
where
\[
      \mathcal{R}^{[p+1]}(s) =  \sum_{k=0}^{p} s^k \, C_{3,k} - \sum_{k=0}^{p-1} \frac{1}{(k+1)!} \, \ad_{W^{[p+1]}}^k \dot{W}^{[p+1]}
\]
and $C_{1,n}, C_{2,n}, C_{3,n}$ are obtained with the recurrence \eqref{recur_C} with $A_1=A_3=A$, $A_2=B$.
 \end{theorem}   

\begin{corollary}
With the same hypothesis as Theorem \ref{th:sym} one has, for $p = 3,5,7, \ldots$,
\[
   \left\| \e^{\frac{t}{2} A} \, \e^{t B} \,  \e^{\frac{t}{2} A} - \e^{W^{[p]}(tA,tB)} \right\|  \le  \sum_{j=p+2}^{p^2 +p -2} t^j \, S_j^{[p]}(A,B),
\]
where $S_j^{[p]}(A,B)$ are linear combinations of norms of nested commutators involving $j$ operators $A$ and $B$.
\end{corollary}

The explicit expression for $p=3$ contains terms up to degree 10 in $t$ and can be found at the website 
\begin{center}
  \url{http://www.gicas.uji.es/Research/bch.html}
\end{center}  
and contains 45 terms. In general, the bounds thus obtained contain even and odd
powers of $t$. The expression of the first term in the error bound, $S_{p+2}^{[p]}$ for $p=3,5,7$ is also available there in the Hall basis.

\section{Error bounds for the Zassenhaus formula} 
 
 The successive operators $C_n(A,B)$ in the infinite product \eqref{zass_1} can be determined by the recursive procedure presented in \cite{casas12eco},
 which, for completeness, is summarized as Algorithm \ref{alg_zass}.
 
\begin{algorithm}
\caption{Computation of the coefficients $C_n$ in the Zassenhaus formula}\label{alg_zass}
\begin{algorithmic}[1]

\State Compute
\[
f_{1,k} =
\sum_{j=1}^{k}
\frac{(-1)^k}{j!(k-j)!}
\,\ad_B^{\,k-j}\ad_A^{\,j} B \qquad\qquad \mbox{ for } \quad k \ge 1
\]

\State Compute
\[
C_n = \frac{1}{n} f_{1,n-1} \qquad\qquad \mbox{ for } \quad n=2,3,4
\]

\State Compute recursively
\[
f_{n,k}=
\sum_{j=0}^{\lfloor k/n \rfloor -1}
\frac{(-1)^j}{j!}
\ad_{C_n}^{\,j}
f_{n-1,k-nj} \qquad\qquad \mbox{ for } \quad n \ge 5, \; k \ge n
\]

\State Determine
\[
C_n =
\frac{1}{n}
f_{\lfloor (n-1)/2 \rfloor ,\,n-1} \qquad\qquad \mbox{ for } \quad n \ge 5
\]

\end{algorithmic}
\end{algorithm}

Since the the infinite product \eqref{zass_1} has to be truncated in practice,  it makes sense to analyze the error associated with
this truncation. More specifically, and following the same approach as for the truncated BCH formula, we next address the following

\

\noindent
\textbf{Problem 2:} \textit{Given the skew-adjoint operators $A$ and $B$, and the Zassenhaus formula after $q \ge 2$ terms, obtain rigorous bounds for 
\begin{equation} \label{bound_z1}
   \mathcal{E}_z \equiv \left\| \e^{t (A+B)} - \e^{t A} \, \e^{t B} \, \e^{t^2 C_2(A,B)} \, \cdots \, \e^{t^q C_q(A,B)} \right\|, \qquad t \in \mathbb{R},
\end{equation}
 in terms of commutators of $A$ and $B$.}
 
 \

Two remarks are in point here: first, when all the $C_k$ in \eqref{bound_z1} are zero we recover again the Lie--Trotter approximation and the bound
\eqref{lie_trotter}. Second, since the spectral norm is unitarily invariant,  it is true that 
\begin{equation} \label{bound_z2}
   \mathcal{E}_z  = \left\| \e^{-t B} \, \e^{-t A} \, \e^{t (A+B)} -  \e^{t^2 C_2(A,B)} \, \cdots \, \e^{t^q C_q(A,B)} \right\|
\end{equation} 
and this is the object we will analyze in the sequel. 

Proceeding in an analogous way as in the truncated BCH formula, we introduce the operators
\[
   F_1(t) = \e^{-t B} \, \e^{-t A} \, \e^{t (A+B)}, \qquad\quad \mbox{ and } \quad\qquad F_2(t) =  \e^{t^2 C_2(A,B)} \, \cdots \, \e^{t^q C_q(A,B)} 
\]
verifying \eqref{ivp_2op} with
\begin{equation} \label{bound_z3}
 \begin{aligned}
  & M_1(t) = -B + \e^{-t \, \ad_B} \, \e^{-t \, \ad_A} B, \\
  & M_2(t) =  2 \, t \, C_2 + \sum_{j=3}^q j \, t^{j-1} \, \e^{t^2 \, \ad_{C_2}} \, \cdots \, \e^{t^{j-1} \, \ad_{C_{j-1}}} C_j,
 \end{aligned}
\end{equation}  
respectively. The desired bounds are then obtained by appropriately expanding $M_1(t)$, $M_2(t)$ in powers of $t$ and applying Lemma \ref{lemma_estimate}.
Thus, for $q=2$ one gets
\[
 M_1(t) = - t \, \ad_A B - t^2 \, G_1(t, -B, \ad_A B) + t^2 \, \e^{-t \ad_B} G_2(t, -A, B),
\]
where the remainders $G_n$ are defined by \eqref{eq:Gn}). In this way,
\[
 M_1(t) - M_2(t)  = - t^2 \, G_1(t, -B, \ad_A B) + t^2 \, \e^{-t \ad_B} G_2(t, -A, B),
\]
so that
\[
\begin{aligned}
  & \| M_1(t) - M_2(t) \|  \le t^2 \|G_1(t, -B, \ad_A B)\| + t^2 \| G_2(t, -A, B) \| \\
  		&	\qquad	 \le t^2 \| \ad_{-B} \ad_A B\| + \frac{t^2}{2} \| \ad_A^2 B\| 
 				  \le t^2 \|[B,[A,B]]\| + \frac{t^2}{2} \|[A,[A,B]]\|
\end{aligned}
\]
and finally
\[
 \left\| \e^{-t B} \, \e^{-t A} \, \e^{t (A+B)} -  \e^{t^2 C_2(A,B)} \right\| \le \frac{t^3}{3} \|[B,[A,B]]\| +  \frac{t^3}{6} \|[A,[A,B]]\|,
\]
thus leading to \eqref{zass_b2} when $t=1$. In view of the expression for $C_3$, eq. \eqref{first_terms_z}, this constitutes an optimal bound: 
it is not possible to get smaller coefficients in front of the respective commutators.

The general expression for any $q \ge 3$ can be obtained by expanding $M_1$ and $M_2$ up to terms in $t^q$ as follows. On the one hand,
\[
  M_1(t) = M_{1,0}(t) + M_{1,1}(t),
\]
with
\begin{equation} \label{expre_M1}
\begin{aligned}
 M_{1,0}(t) & = \sum_{n=1}^{q-1} t^n \sum_{j=0}^n \frac{1}{j! (n-j)!} \ad_{_B}^j \ad_{-A}^{n-j} B  \\
 M_{1,1}(t) & =  t^q \sum_{n=1}^{q-1} \frac{1}{n!} G_{q-n}(t, -B, \ad_{-A}^n B) + t^q \, \e^{t \ad_{-B}} \, G_q(t, -A, B).
\end{aligned}
\end{equation}   
On the other hand, if we introduce the operators $Z_j(t)$, $j=2,\ldots, q$, defined recursively through
\begin{equation} \label{zetas_z}
  \begin{aligned}
  & Z_q(t) = q t^{q-1} C_q \\
  & Z_j(t) = j t^{j-1} C_j + \e^{t^j \ad_{C_j}} Z_{j+1}(t), \qquad\quad j= q-1, q-2, \ldots, 2,
 \end{aligned} 
\end{equation}
then it is clear that $M_2(t) = Z_2(t)$. In fact, if we write 
\[
  Z_j(t) = Z_{j,0}(t) + Z_{j,1}(t), \qquad j=2,\ldots, q
\]
with  $Z_{q,0}(t) =  q \, t^{q-1} C_q$, then 
\begin{equation} \label{expre_Z1}
  \begin{aligned}
   & Z_{j,0}(t) = j \, t^{j-1} C_j + \sum_{k=0}^{\ell_j -1} \frac{t^{j k}}{k!} \ad_{C_j}^k Z_{j+1,0}(t), \qquad  j=q-1, \ldots, 2 \\
   & Z_{j,1}(t) =  t^{j \ell_j} \, G_{\ell_j}(t^j, C_j, Z_{j+1,0}(t)) + \e^{t^j \ad_{C_j}} Z_{j+1,1}(t), \qquad \ell_j := \lfloor \frac{q}{j} \rfloor.
  \end{aligned} 
\end{equation}
Then, by the way the Zassenhaus formula is obtined, we have $M_{1,0}(t) = Z_{2,0}(t)$, so that
\[
  M_1(t) - M_2(t) = M_{1,1}(t) - Z_{2,1}(t),
\]
and therefore $\|M_1(t) - M_2(t)\| \le \|M_{1,1}(t)\| + \|Z_{2,1}(t)\|$.  Now, from equations \eqref{expre_M1} and \eqref{norm_G}, it follows that
\[
  \|M_{1,1}(t)\| \le t^q \sum_{n=1}^{q-1} \frac{1}{n! (q-n)!} \|\ad_B^{q-n} \ad_{A}^n B\| +  \frac{t^q}{q!} \, \|\ad_A^q B\|,
\]  
whereas \eqref{expre_Z1} leads to
\[
\begin{aligned}
 \|Z_{2,1}(t)\| & \le \| t^{2 \ell_2} G_{\ell_2}(t^2, C_2, Z_{3,0} + \e^{t^2 \ad_{C_2}} Z_{3,1}\| \le t^{2 \ell_2} \|G_{\ell_2}(t^2, C_2, Z_{3,0}\| + \|Z_{3,1}\| \\
 & \le \sum_{j=2}^{q-1} t^{j \ell_j} \|G_{\ell_j}(t^2, C_j, Z_{j+1,0}\| \le \sum_{j=2}^{q-1} \frac{t^{j \ell_j}}{\ell_j!} \|\ad_{C_j}^{\ell_j} Z_{j+1,0}\|.
\end{aligned}
\]
In consequence, 
\[ 
\aligned 
&  \left\| \e^{t (A+B)}  - \e^{t A} \, \e^{t B} \, \e^{t^2 C_2(A,B)} \, \cdots \, \e^{t^q C_q(A,B)} \right\|
\le   
\frac{t^{q+1}}{q+1}\sum_{n=1}^{q-1}  \frac{1}{n!(q-n)!}\,\|
\operatorname{ad}^{q-n}_{B} \operatorname{ad}_{A}^{n}B \|
\\& \qquad+
\frac{t^{q+1}}{(q+1)!}\, \|\ad^q_{A} B\|  + \sum_{j=2}^{q-1}\int_0^t   \frac{s^{j \ell_j}}{\ell_j!} \|\ad_{C_j}^{\ell_j} Z_{j+1,0}(s)\| \, ds.
\endaligned
\] 
If we particularize this recursion to the case $q=3$, we get   
\[
  \|M_1(t) - M_2(t)\|  \le \frac{t^3}{2} \|\ad_B^2 \ad_A B\|  + \frac{t^3}{2} \|\ad_B \ad_A^2 B\|  + \frac{t^3}{6} \|\ad_A^3 B\| 
  + 3 t^4 \|\ad_{C_2} C_3\|
\] 
and therefore
\[
\begin{aligned}
 &  \left\| \e^{t (A+B)} -  \e^{t A} \, \e^{t B} \, \e^{t^2 C_2(A,B)} \,  \e^{t^3 C_3(A,B)} \right\| \le \frac{t^4}{8} \|[B,[B,[A,B]]]\| + \frac{t^4}{8} \|[B,[A,[A,B]]]\| \\
 & \qquad\qquad + \frac{t^4}{24} \|[A,[A,[A,B]]]\| + \frac{3}{5} t^5 \|[C_2,C_3]\|.
\end{aligned} 
\]
In general we have
\begin{equation} \label{gen_expl}
  \left\| \e^{t (A+B)} - \e^{t A} \, \e^{t B} \, \e^{t^2 C_2(A,B)} \, \cdots \, \e^{t^q C_q(A,B)} \right\|
   \le \sum_{j=q+1}^{2q-1} t^j D_j^{[q]}(A,B), \qquad q \ge 2,
\end{equation}
where each $D_j^{[q]}(A,B)$ is a linear combination of norms of nested commutators involving $j$ operators $A$ and $B$.

The interested reader can find the code implementing this recursion, together with an explicit formula for the general bound \eqref{gen_expl}
and several examples of application at the website
\begin{center}
  \url{http://www.gicas.uji.es/Research/bch.html}
\end{center}

\section{Concluding remarks}
\label{conclusions}

The Baker--Campbell--Hausdorff formula is a fundamental mathematical object with a wide range of applications in quantum mechanics, numerical 
analysis of differential equations, control theory, differential geometry and mathematical physics. 
It is therefore not surprising that it has attracted so much attention over the years, with recent contributions emphasizing the obtention of closed form
expressions in particular settings \cite{moodie21aep}.

In many practical applications, however, the BCH formula must be truncated, so it is meaningful to establish rigorous bounds on the error introduced by such truncation. This is particularly true in the case of the simulation of Hamiltonian systems by means of quantum algorithms (for example, through Trotter--Suzuki-type compositions). In that case, these bounds make it possible to obtain information about the number of circuit gates or about how the error scales with time.


The Zassenhaus formula is also highly useful in quantum Hamiltonian simulation, as it naturally captures the algebraic interplay between non-commuting operators through a hierarchy of correction terms, each represented by the exponential of increasingly nested commutators. In many cases, particularly for weakly coupled systems, the magnitude of these higher-order corrections decreases rapidly, so that a suitable truncation can yield sufficiently accurate results for local Hamiltonians \cite{nguyen25zei}. In this context, sharp bounds for the error incurred by truncating the expansion can therefore be of considerable practical relevance.

In this work we have carried out a systematic treatment of this issue in a finite-dimensional setting, 
establishing general bounds then the BCH and the Zassenhaus formulas are
truncated in the unitary case, when the operators appearing in the exponents are skew-adjoint. We have also analyzed for completeness the symmetric
BCH formula. These bounds are polynomials in the parameter
involved whose coefficients are iterated commutators of the skew-adjoint operators, thus exhibiting the commutator scaling property. The provided results can be implemented in a computer algebra systems to render explicit expressions up to high values of the truncation index $m$. Although, for illustration purposes, we present these expressions using the classical Hall basis, other bases of the corresponding free Lie algebra may equally be employed. It should be noted, however, that the resulting bounds depend on the particular choice of basis taken.
The treatment can also be extended to other variants of the BCH \cite{vanbrunt16str} and Zassenhaus formulas \cite{arnal17ots}.

\subsection*{Acknowledgements}
This work has been supported by 
Ministerio de Ciencia e Innovaci\'on (Spain) through project  PID2022-136585NB-C21, 
MCIN/AEI/10.13039/501100011033/FEDER, UE. The authors wish to express their gratitude to A. Murua 
for making his \emph{Mathematica} package \texttt{LieSeries} available to them.

\subsection*{Compliance with Ethical Standards}

All authors declare that they have no conflicts of interest.

 \appendix
 
 \section{Appendix: explicit expressions of error bounds for the BCH formula }

We next collect the bound obtained from Theorem  \ref{th:general_bound} for the BCH formula when the first three terms in the series is considered,
$m=3$, in the classical Hall basis. Specifically,
\[
 Z^{[3]}(tA, tB) = t(A+B) + \frac{t^2}{2} [A, B] + \frac{t^3}{12} ([A,[A,B]] - [B,[A,B]]).
\]
Then we get
\begin{equation} \label{B1_3}
    \left\| \e^{t A} \, \e^{t B} - \e^{Z^{[3]}(tA, tB) } \right\|  \le B_1^{[3]} \equiv \sum_{j=4}^8 t^j C_j^{[3]},
\end{equation}
with
\allowdisplaybreaks
\begin{align*}
 & C_4^{[3]} = \frac{1}{48} \big\| -[A,[A,[A,B]]] + [B,[B,[A,B]]] \big\| + \frac{1}{24} \|[A,[A,[A,B]]]\| \\
  & \qquad + \frac{1}{48} \big\| [A,[A,[A,B]]] + 2 [B,[A,[A,B]]] + [B,[B,[A,B]]] \big\| \\
 & C_5^{[3]} =  \frac{1}{240} \big\| -[[A,B],[A,[A,B]]] + [[A,B],[B,[A,B]]] \big\| +  \big\| \frac{1}{180} \Big([A,[A,[A,[A,B]]]] \\
   & \qquad + [B,[A,[A,[A,B]]]] - [B,[B,[A,[A,B]]]] - [B,[B,[B,[A,B]]]] \Big) \\
   & \qquad + \frac{1}{120}  [[A,B],[A,[A,B]]] + \frac{1}{360} [[A,B],[B,[A,B]]] \big\| \\
 & C_6^{[3]} =  \frac{1}{288}  \big\|[[A,B],[A,[A,[A,B]]]]\| - \|[[A,B],[B,[B,[A,B]]]] \big\| \\
 & C_7^{[3]} =  \frac{1}{2016}  \big\|[[A,B],[[A,B],[A,[A,B]]]]\| - \|[[A,B],[[A,B],[B,[A,B]]]] \big\| \\
    & \qquad + \frac{1}{3024} \big\| [[A,[A,B]], [A,[A,[A,B]]]] - \|[[A,[A,B]], [B,[B,[A,B]]]]  \\
    & \qquad  - [[B,[A,B]], [A,[A,[A,B]]]] + [[B,[A,B]], [B,[B,[A,B]]]] \big\|\\
 & C_8^{[3]} = \frac{1}{13824} \big\|[[A,[A,B]], [[A,B],[A,[A,B]]]]  \\
    & \qquad\quad -  [[A,[A,B]], [[A,B],[B,[A,B]]]] \\
    & \qquad\quad - [[B,[A,B]], [[A,B],[A,[A,B]]]] \\
    & \qquad\quad + [[B,[A,B]], [[A,B],[B,[A,B]]]] \big\|
\end{align*}
 
 
\bibliographystyle{siam}


\end{document}